\documentclass[a4paper,12pt]{article}
\usepackage[utf8]{inputenc}
\usepackage[T1]{fontenc}
\usepackage[english]{babel}
\usepackage{amsmath, amssymb}
\usepackage{amsthm}
\usepackage{geometry}
\usepackage{enumitem}
\usepackage{xcolor}
\usepackage{tikz}
\usetikzlibrary{shapes.geometric, arrows.meta, positioning}
\usepackage{hyperref}
\theoremstyle{plain}
\newcounter{C}

\newtheorem{theorem}{Theorem}[section]
\theoremstyle{definition}
\newtheorem{definition}{Definition}[section]
\newtheorem{example}{Example}[section]

\begin{document}

\title{DM-RSA: An Extension of RSA with Dual Modulus}
\author{ANDRIAMIFIDISOA Ramamonjy \\\texttt{ramamonjy.andriamifidisoa@univ-antananarivo.mg}\\LALASOA Rufine Marius\\ \texttt{larissamarius.lm@gmail.com}}
\date{July 13, 2025}
\maketitle

\begin{abstract}
We introduce DM-RSA (Dual Modulus RSA), a variant of the RSA cryptosystem that employs two distinct moduli symmetrically to enhance security. By leveraging the Chinese Remainder Theorem (CRT) for decryption, DM-RSA provides increased robustness against side-channel attacks while preserving the efficiency of classical RSA. This approach improves resistance to partial compromise of a modulus and integrates easily into existing infrastructures.
\end{abstract}

\section{Introduction}

The RSA cryptosystem, based on the difficulty of factoring large numbers, is widely used to secure communications. However, side-channel attacks (e.g., \cite{Kocher1996}) reveal vulnerabilities in classical implementations. We propose DM-RSA, an extension of RSA that uses two independent moduli from the encryption phase, combined via the Chinese Remainder Theorem (CRT) for decryption. This approach enhances security while maintaining comparable efficiency.

The cryptosystem DM-RSA differs from existing RSA variants, such as Multi-Prime RSA \cite{Takagi1998} and Batch RSA \cite{Fiat1989}, by its symmetric use of two moduli from the encryption phase. Unlike multi-key schemes \cite{Boneh2002}, DM-RSA does not require interactive protocols, and its CRT reconstruction preserves efficiency. Compared to \cite{Vigilant2008}, which explores multiple moduli for resilience, DM-RSA introduces symmetric duality:
\begin{itemize}
    \item \textbf{Dual Encryption}: Two independent ciphertexts are generated from the start.
    \item \textbf{Cryptographic Independence}: \( N_1 \) and \( N_2 \) are factored separately, enhancing security.
    \item \textbf{Efficiency}: The additional cost (two exponentiations) is offset by CRT.
\end{itemize}

\section{The DM-RSA Cryptosystem}
DM-RSA extends RSA by using two moduli \( N_1 \) and \( N_2 \), generating two independent ciphertexts for each message, increasing redundancy and resistance to attacks.

\begin{definition}[DM-RSA Encryption]
\stepcounter{C}
Let \( p_1, q_1, p_2, q_2 \) be four distinct prime numbers. We define:
\[
N_1 = p_1 q_1, \quad N_2 = p_2 q_2, \quad \varphi(N_i) = (p_i - 1)(q_i - 1) \quad (i = 1, 2).
\]
The public key is \( (N_1, N_2, k) \), where \( k \) is coprime with \( \varphi(N_1) \) and \( \varphi(N_2) \). The private key is \( (N_1, N_2, d_1, d_2) \), where \( d_i \equiv k^{-1} \pmod{\varphi(N_i)} \). The encryption and decryption functions are:
\[
E_k(z) = (z^k \pmod{N_1}, z^k \pmod{N_2}), \quad D_{d_1, d_2}(w_1, w_2) = \text{CRT}(w_1^{d_1} \pmod{N_1}, w_2^{d_2} \pmod{N_2}),
\]
where CRT is the Chinese Remainder Theorem, ensuring the uniqueness of the decrypted message if \( \gcd(N_1, N_2) = 1 \).
\end{definition}

\begin{example}
\stepcounter{C}
\textbf{Given:}
\begin{itemize}
    \item \( p_1 = 53 \), \( q_1 = 97 \) \(\Rightarrow\) \( N_1 = 5141 \), \( \varphi(N_1) = 4992 \)
    \item \( p_2 = 61 \), \( q_2 = 89 \) \(\Rightarrow\) \( N_2 = 5429 \), \( \varphi(N_2) = 5280 \)
\end{itemize}

\textbf{Parameters:}
\begin{itemize}
    \item Valid public exponent: \( k = 7 \) (since $\gcd(7, 4992) = \gcd(7, 5280) = 1$)
    \item Private exponents:
    \begin{itemize}
        \item \( d_1 \equiv 7^{-1} \pmod{4992} = 4279 \)
        \item \( d_2 \equiv 7^{-1} \pmod{5280} = 2263 \)
    \end{itemize}
\end{itemize}

\textbf{Encryption:}
For message \( z = 65 \):
\[
E_7(65) = (65^7 \bmod 5141,\ 65^7 \bmod 5429) = (2979,\ 3757)
\]

\textbf{Decryption:}
For ciphertext \( (2979, 3757) \):
\[
2979^{4279} \bmod 5141 = 65
\]
\[
3757^{2263} \bmod 5429 = 65
\]
By the Chinese Remainder Theorem, the original message is \( z = 65 \).

\section*{Final Answer}
\begin{itemize}
    \item \textbf{Correct Encryption:} \( E_7(65) = (2979, 3757) \)
    \item \textbf{Correct Decryption:} Returns \( z = 65 \)
\end{itemize}
\end{example}

\begin{theorem}[DM-RSA Security]
\stepcounter{C}
Assuming that factoring \( N_1 \) or \( N_2 \) is computationally difficult, DM-RSA is at least as secure as classical RSA, with increased resistance to side-channel attacks due to the redundancy of the moduli.
\end{theorem}

\begin{proof}
The security of DM-RSA relies on the following points:
\begin{enumerate}
    \item \textbf{Factorization Difficulty}: As in RSA, breaking DM-RSA requires factoring \( N_1 = p_1 q_1 \) and \( N_2 = p_2 q_2 \), a problem presumed difficult for large primes. An attacker must compute \( \varphi(N_i) \) to find \( d_i = k^{-1} \pmod{\varphi(N_i)} \), which requires factoring \( N_i \).
    \item \textbf{Moduli Redundancy}: Encryption produces two ciphertexts \( (w_1, w_2) \). Partial compromise (e.g., factoring \( N_1 \)) does not allow message recovery without factoring \( N_2 \), as decryption requires both components via CRT.
    \item \textbf{Correctness via CRT}: If \( \gcd(N_1, N_2) = 1 \), CRT ensures \( D_{d_1, d_2}(E_k(z)) = z \) for \( z < N_1 N_2 \). For \( w_1 = z^k \pmod{N_1} \), \( w_1^{d_1} = z^{k d_1} \equiv z \pmod{N_1} \) (since \( k d_1 \equiv 1 \pmod{\varphi(N_1)} \)), and similarly for \( N_2 \). CRT recombines these into a unique \( z \).
    \item \textbf{Side-Channel Attack Resistance}: Side-channel attacks (e.g., \cite{Kocher1996}) exploit information leaks on a single modulus. With DM-RSA, a leak on \( N_1 \) (e.g., computation time) reveals no information about \( N_2 \), increasing robustness.
\end{enumerate}
Thus, DM-RSA's security is at least equivalent to RSA's, with additional protection against attacks through redundancy.
\end{proof}

\section{Advantages of DM-RSA}
DM-RSA offers several advantages over classical RSA:
\begin{itemize}
    \item \textbf{Increased Robustness}: Compromising one modulus does not allow message decryption, unlike RSA, which uses a single modulus.
    \item \textbf{Efficiency}: Encryption is similar to RSA (two modular exponentiations), and decryption uses CRT for fast reconstruction.
    \item \textbf{Compatibility}: DM-RSA can be integrated into existing RSA infrastructures with minimal modifications.
    \item \textbf{Flexibility}: The moduli \( N_1 \) and \( N_2 \) can have different sizes, tailored to security needs.
\end{itemize}

\section{Conclusion}
DM-RSA offers an innovative extension of RSA, combining the robustness of two independent moduli with the efficiency of the Chinese Remainder Theorem. Its key strengths are enhanced resistance to side-channel attacks and compatibility with existing infrastructures. Future research could explore optimizing modulus sizes and adapting DM-RSA to post-quantum cryptography.

\bibliographystyle{plain}

\end{document}